\documentclass[conference]{IEEEtran}
\IEEEoverridecommandlockouts
\usepackage{cite}
\usepackage{amsmath,amssymb,amsfonts}
\usepackage{amsthm}
\usepackage{mathtools}
\usepackage{algorithm}
\usepackage{algpseudocode}
\usepackage{yfonts}
\usepackage{graphicx}
\usepackage{textcomp}
\usepackage[dvipsnames]{xcolor}
\usepackage{svg}
\usepackage{import}
\usepackage{tikz}
\usepackage{pgfplots}
\pgfplotsset{compat=newest}
\usepackage{pgfplotstable}
\usepackage{comment}
\usepackage{environ}
\usepackage{enumitem}
\usepackage{framed}
\graphicspath{{figs/}}
\def\BibTeX{{\rm B\kern-.05em{\sc i\kern-.025em b}\kern-.08em
    T\kern-.1667em\lower.7ex\hbox{E}\kern-.125emX}}

\DeclareUnicodeCharacter{2212}{\textendash}

\DeclareMathOperator{\lcm}{lcm}
\renewcommand{\mod}{\mathbin{\%}}
\newcommand{\cmod}{\mathbin{\overline{\%}}}
\newcommand{\s}[1]{\left\langle #1 \vphantom{\frac{1}{2}} \right\rangle}
\newcommand{\ts}[3]{\bigcup_{#2}^{#3} \left\{ \vphantom{M^{M^M}_M} #1 \right\}  }
\newcommand{\seq}[1]{(#1_k)}
\renewcommand{\r}[1]{\mathfrak{#1}}
\renewcommand{\a}{\alpha}
\renewcommand{\aa}{\dot{\alpha}}
\renewcommand{\b}{\beta}
\newcommand{\n}{\eta}
\newcommand{\E}[1]{   \mathbb{E}  \mathopen{}\left[   #1 \right] \mathclose{}  }
\newcommand{\hE}[1]{\hat{\mathbb{E}} \mathopen{}\left[  #1 \right]  \mathclose{} }
\renewcommand{\t}[1]{\tilde{#1}}
\newcommand{\floor}[1]{\left\lfloor #1 \right\rfloor}
\newcommand{\ceil}[1]{\left\lceil #1 \right\rceil}

\newtheorem{lemma}{Lemma}
\newtheorem{theorem}{Theorem}
\newtheorem{corollary}{Corollary}

\renewcommand{\figurename}{Fig.}
\newcommand{\figref}[1]{\figurename~\ref{#1}}

\makeatletter
\newcommand*{\centerfloat}{%
  \parindent \z@
  \leftskip \z@ \@plus 1fil \@minus \textwidth
  \rightskip\leftskip
  \parfillskip \z@skip}
\makeatother

\begin{document}

\title{Age-of-Information in Clocked Networks
\thanks{
RIchard Schöffauer is supported by the DFG (SPP 1914). Gerhard Wunder is supported by both the DFG (SPP 1914) and the BMBF (6G-RIC).}
}

\author{\IEEEauthorblockN{RIchard Schöffauer}
\IEEEauthorblockA{\textit{Dept. of Mathematics and Computer Science} \\
\textit{Freie Universität Berlin}\\
Berlin, Germany \\
richard.schoeffauer@fu-berlin.de}
\and
\IEEEauthorblockN{Gerhard Wunder}
\IEEEauthorblockA{\textit{Dept. of Mathematics and Computer Science} \\
\textit{Freie Universität Berlin}\\
Berlin, Germany \\
gerhard.wunder@fu-berlin.de}
}

\maketitle

\begin{abstract}
We derive key features of the Age-of-Information distribution in a system whose activities are strictly limited to periodic instances on a global time grid. In particular, one agent periodically generates updates while the other agent periodically uses the most recently received of those updates. Likewise, transmission of those updates over a network can only occur periodically. All periods may differ. We derive results for two different models: a basic one in which the mathematical problems can be handled directly and an extended model which, among others, can also account for stochastic transmission failure, making the results applicable to instances with wireless communication.
For both models, a suitable approximation for the expected Age-of-Information and an upper bound for its largest occurring value are developed. For the extended model (which is the more relevant one from a practical standpoint) we also present numerical results for the distribution of the approximation error for numerous parameter choices. 
\end{abstract}


\section{Introduction \& Related Research}

Timely delivery of data is an important feature of communication systems. To quantify this feature, the Communication-Delay (ComDelay) has long been the predominant metric. In recent years however, especially in the context of machine-to-machine communication (which includes many communication scenarios from IoT), the Age-of-Information (AoI) metric has gained considerable attention. While ComDelay captures the elapsed time between transmission and reception (and thus only considers the communication infrastructure), AoI measures how "old" the current information at the receiver really is. To disclose the difference, imagine a communication line with a 1 second ComDelay and assume that a packet of information is send over that line \textit{only} every full hour. The ComDelay is not influenced by this usage and remains 1 second; however the AoI at the receiver is about 30 minutes on average. Every reception, it starts from 1 second and linearly grows with time, peaking at 1 hour and 1 second just before the next reception.

So far, the AoI metric has been studied in a variety of scenarios, an overview over which is given a few lines further down. However, all these scenarios are concerned with stochastic systems: data generation is irregular, intermediate processing takes a random amount of time, transmission is defective. While there certainly is a lot of motivation for such scenarios, it is surprising that rather simple and deterministic set ups have been ignored until now. In practice, many autonomous systems do engage in their tasks in a very regular fashion: measuring system states, processing new outputs, applying corrective actions, they all take place in a predefined amount of time and are repeated thereafter periodically. Even a communication network itself usually schedules transmissions in predefined resource blocks over a time grid.

\begin{figure}
    \centering
    \footnotesize
\begingroup%
  \makeatletter%
  \providecommand\color[2][]{%
    \errmessage{(Inkscape) Color is used for the text in Inkscape, but the package 'color.sty' is not loaded}%
    \renewcommand\color[2][]{}%
  }%
  \providecommand\transparent[1]{%
    \errmessage{(Inkscape) Transparency is used (non-zero) for the text in Inkscape, but the package 'transparent.sty' is not loaded}%
    \renewcommand\transparent[1]{}%
  }%
  \providecommand\rotatebox[2]{#2}%
  \newcommand*\fsize{\dimexpr\f@size pt\relax}%
  \newcommand*\lineheight[1]{\fontsize{\fsize}{#1\fsize}\selectfont}%
  \ifx\svgwidth\undefined%
    \setlength{\unitlength}{230.52467406bp}%
    \ifx\svgscale\undefined%
      \relax%
    \else%
      \setlength{\unitlength}{\unitlength * \real{\svgscale}}%
    \fi%
  \else%
    \setlength{\unitlength}{\svgwidth}%
  \fi%
  \global\let\svgwidth\undefined%
  \global\let\svgscale\undefined%
  \makeatother%
  \begin{picture}(1,0.91447813)%
    \lineheight{1}%
    \setlength\tabcolsep{0pt}%
    \put(0,0){\includegraphics[width=\unitlength,page=1]{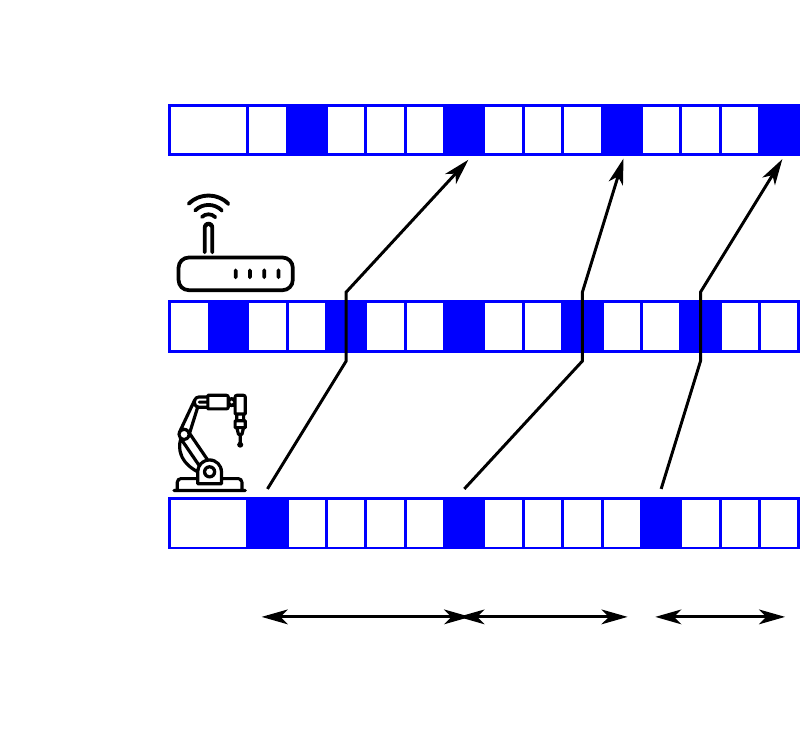}}%
    \put(0.62926068,0.00852127){\makebox(0,0)[t]{\lineheight{1.25}\smash{\begin{tabular}[t]{c}Resulting AoI\end{tabular}}}}%
    \put(0,0){\includegraphics[width=\unitlength,page=2]{set_up.pdf}}%
    \put(0.08821482,0.74631101){\makebox(0,0)[t]{\lineheight{1.25}\smash{\begin{tabular}[t]{c}Agent A\end{tabular}}}}%
    \put(0.08821482,0.5003811){\makebox(0,0)[t]{\lineheight{1.25}\smash{\begin{tabular}[t]{c}Network N\end{tabular}}}}%
    \put(0.08821482,0.25445118){\makebox(0,0)[t]{\lineheight{1.25}\smash{\begin{tabular}[t]{c}Agent B\end{tabular}}}}%
    \put(0,0){\includegraphics[width=\unitlength,page=3]{set_up.pdf}}%
    \put(0.23577277,0.00852127){\makebox(0,0)[t]{\lineheight{1.25}\smash{\begin{tabular}[t]{c}Time Grid\end{tabular}}}}%
    \put(0,0){\includegraphics[width=\unitlength,page=4]{set_up.pdf}}%
  \end{picture}%
\endgroup%

    \caption{General set-up: two agents communicating over a network. Each unit features its independent working cycle (distance between blue slots) on a common time grid. Agent B periodically generates status-updates which are send over a network N to agent A. Agent A can only make us of this new information at the start of its next processing cycle.}
    \label{fig::system_model}
\end{figure}

Therefore, in this paper, we investigate an elementary system that consists of two subsystems with periodic workflow and a communication network (likewise with periodical workflow) over which data is exchanged. It turns out that even without any stochastic effects taken into account, the non-linear behavior of the AoI makes it difficult to yield analytical expressions for e.g. the average AoI or the largest occurring AoI in the system. Our \textbf{contribution} consists of deriving expressions for exactly such key performance parameters in connected subsystems with periodic workflow. Along the way, we state structural results regarding the AoI process over time for those systems.

Regarding \textbf{related research}, note that to the best of our knowledge, we are the first to investigate the AoI in the context of a deterministic framework. Therefore the following references are only of limited use in order to contextualize our contribution.
Investigation of the AoI metric started with \cite{Kaul2011, Kaul2012} where the average AoI was optimized over the update generation rate for elementary queueing systems such as M/M/1, M/D/1, and D/M/1.
Naturally, this sparked further research exploring various other queueing systems and other AoI features like the peak AoI.
The benefit of having multiple servers (e.g. M/M/2 queue) was investigated in \cite{Kam2016,Costa2014,Yates2018a,Bedewy2016}.
The effect of packet deadlines, after which packets are deleted, freeing up network resources in the process, was demonstrated in \cite{Kam2016a,Kam2016b,Inoue2018}.
Considering that update generation and or transmission is resource consuming can be modeled by energy harvesting sources which, in the context of AoI, was investigated in \cite{Yates2015,Bacinoglu2017,Arafa2017,Wu2018,Arafa2020}. Direct power constraints on the other hand were considered in \cite{Tang2019}.
AoI in the context of wireless networks, especially broadcast networks, where link transmission is error prone and interference constraints limit the set of admissible scheduling policies was investigated in \cite{Kadota2018,Kadota2019, Talak2020, Talak2018d}. Here, the focus lies on finding (near-) optimal policies to minimize the AoI for the entire network. The special case of stochastic policies that trigger transmissions over individual links based on Bernoulli processes (ALOHA-like policies) was researched in \cite{Talak2018a, Yates2017}.
AoI in multi hop wireless networks was dealt with in \cite{Talak2018b, Farazi2019a, Buyukates2019},
and for the special case of gossip networks in \cite{Yates2021,Chaintreau2009,Selen2013}.

We close the introduction with some notes on the \textbf{notation}:
\begin{itemize}[leftmargin=*]
    \item For sandwiched values we use the following shorthand:
    \begin{equation}
        z = x \pm y
        \qquad \Longleftrightarrow \qquad
        x-y \leq z \leq x+y
    \end{equation}
    \item All sets in this paper are multisets, i.e. a single element might have more than one occurrence. (Therefore a multiset might represent a distribution of values.)
    \item Sets of evenly distributed natural numbers are denoted as the triple $\langle$start, step-size, number of steps$\rangle$:
    \begin{equation}
    \begin{gathered}
        \s{y,  x ,X} := \Big\{ y, \ y+ x, \ y+2x, \dots y+ (X-1) x \Big\}
    \end{gathered}
    \end{equation}
    \item The modulo-operator will be abbreviated by $\mod$, its complement by $\cmod$:
    \begin{equation}
        \label{eq::mod_to_ceil}
        x \mod y := x - \left\lfloor \frac{x}{y} \right\rfloor y
        \qquad \qquad
        x \cmod y := \left\lceil \frac{x}{y} \right\rceil y - x
    \end{equation}
    (While $x \mod y$ is the distance between $x$ and the closest multiple of $y$ that is still smaller than $x$, $x \cmod y$ is the distance between $x$ and the closest multiple of $y$ that is still larger than $x$. It holds that $(-x) \mod y = x \cmod y$.)
\end{itemize}

\section{Basic System Model}

\begin{figure}
    \centering
    \footnotesize
    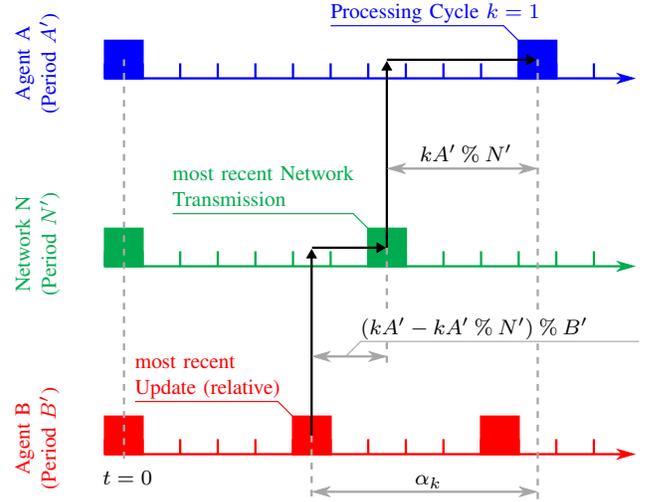
    \caption{Basic model. All periods start at $t=0$ and everything else happens instantly.}
    \label{fig::general_system}
\end{figure}

Our elementary system consists of two agents (A and B) and a communication network (N). Time is slotted. Agent A periodically engages in a processing cycle for which he uses the latest status-update he received from agent B. While agent A is processing, he may receive new status-updates from agent B but is not able to make use of them until the start of the next processing cycle. Only the newest status-update from B at A is needed for processing and thus only the last received status-update is buffered. Still, even this status-update is, in general, several time steps old when it is used at the start of the next processing cycle. This age is called the Age-of-Information, or short AoI.
Agent B generates its status-updates periodically. Likewise, the network N can transmit updates from B to A only at periodically distributed time steps (e.g. because other users require communication resources as well). 

In particular: agent A starts a new processing cycle every $A'$-th time step, agent B generates an update every $B'$-th time step, network N transmits the latest update from agent B every $N'$-th time step.
The relations between $A',B',N'$ (i.e. their common divisors) are of great importance to the problem and are defined through the following compositions:
\begin{equation}
    \label{eq::compositions}
    \begin{aligned}
        A' &= A \cdot b \cdot n
        \\
        B' &= B \cdot b \cdot a
        \\
        N' &= N \cdot n \cdot a
    \end{aligned}
    \qquad \text{such that} \qquad
    \begin{aligned}
        a &= \gcd(B',N')
        \\
        b &= \gcd(A',B')
        \\
        n &= \gcd(A',N')
    \end{aligned}
\end{equation}
All quantities are $\in \mathbb{N}$ (natural numbers). We point out that we do not consider a divisor (except $1$) that divides all three periods $A',B',N'$ since such a divisor can simply be accounted for by scaling the entire time axis.

From here, we develop two different models. The first one is labeled the \textbf{"basic" model}:
\begin{itemize}[leftmargin=*]
    \item Update generation and transmission is instant.
    \item All periods start at $t=0$.
    \item Transmissions are always successful.
\end{itemize}
\figref{fig::system_model} illustrates the setting. (Notably, the resulting system model already exhibits the fundamental mathematical challenges we tackle in this paper, such that most results can be readily extended to the second model).

We are interested in the AoI process $\seq{\a}_{k\in \mathbb{N}}$ (from now on denoted just as $\seq{\a}$) where $\alpha_k$ is the AoI of the latest received update (from B at A) at the beginning of agent A's processing cycle $k$. In case of the basic model, this process can readily be deduced from \figref{fig::system_model} as 
\begin{equation}
\label{eq::basis}
\begin{aligned}
    \alpha_k &=
    kA' \mod N' + \left( kA'  - kA' \mod N' \right) \mod B'
    \\ &=
    kA' \mod N' + \left( kA' \mod B'  - kA' \mod N' \right) \mod B'
\end{aligned}
\end{equation}
Even though $\seq{\a}$ is deterministic, we will interpret $\seq{\a}$ as a stochastic process since its behavior can quickly become quite complex.



As a direct consequence of this kind of modeling, in time step $t=0$, agent B generates a status-update which network N immediately transmits to agent A and which can then be used in agent A's first processing cycle. Hence, in time slot $t=0$ (and periodically thereafter) the AoI is zero, which already motivates a second system model.

\section{Extended System Model}

Intuitively, the AoI at agent A should never truly be zero for several reasons, e.g. due to the time it takes for the transmission to reach agent A or due to the time it takes agent B to generate the status-update. Therefore we will extend our results to a system model which accounts for such delays and also introduces some other aspects that are quite common in this field of research. We refer to this new model as the \textbf{"extended" model}:
\begin{itemize}[leftmargin=*]
    \item Update generation and transmission each take one time step.
    \item Relative to time step $t=0$, the first periods of agent B and network N are delayed by $\Delta_B$ and $\Delta_N$ time steps.
    \item Transmissions succeed only with probability $p$.
\end{itemize}
For a visual representation see \figref{fig::system_extended}.

In this model we denote the AoI process with $\seq{\aa}$ while its (verbal) definition stays the same relative to the model. Compared to the basic model, the elapsed time since the last network transmission relative to the current time step $t=kA'$ (the beginning of agent A's $k$-th processing cycle) changes according to
\begin{gather}
    \label{eq::der_1}
    \hphantom{\rightsquigarrow} \quad kA' \mod N'
    \\
    \label{eq::der_2}
    \rightsquigarrow \quad (kA' - \Delta_N) \mod N'
    \\
    \label{eq::der_3}
    \rightsquigarrow \quad 1 + (kA' - \Delta_N - 1) \mod N'
    \\
    \label{eq::der_4}
    \rightsquigarrow \quad lN' + 1 + (kA' - \Delta_N - 1) \mod N'
\end{gather}

\begin{figure}
    \centering
    \footnotesize
    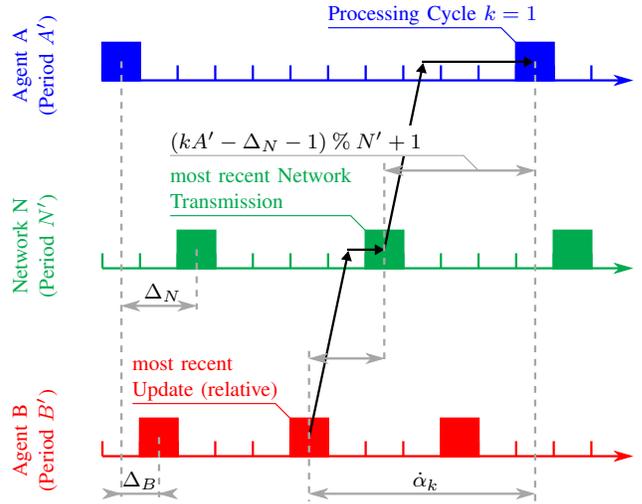
    \caption{Extended model. Start of initial periods are shifted, update generation and transmission takes up one time step.}
    \label{fig::system_extended}
\end{figure}

The first change follows from the shift of network N's period.
The second change follows from the transmission delay: \eqref{eq::der_2} and \eqref{eq::der_3} are identical if the value of \eqref{eq::der_2} is from the set $\{1,\dots N'-1\}$. For these values the transmission delay does not effect the outcome since there is enough time until agent A's processing cycle starts. However, if transmission takes place in the same time step as A's processing cycle starts (eq. \eqref{eq::der_2} yields $0$), then \textit{in the extended} model, agent A cannot make use of the new update yet and has to resort to the previous transmission which is $N'$ steps in the past (eq. \eqref{eq::der_3} yields $N'$).
The third change follows from unsuccessful transmissions. If the latest $l$ transmissions did not succeed, \eqref{eq::der_4} accounts for that by adding $lN$ time steps on top of \eqref{eq::der_3}. We will account for the probability later.

On agent B's side, the elapsed time since the last status-update from agent B that the network could use for transmission becomes
\begin{gather}
    \hphantom{\rightsquigarrow}\quad (kA' - \bigstar) \mod B'
    \\
    \rightsquigarrow \quad  (kA' - \bigstar - \Delta_B) \mod B'
    \\
    \label{eq::derr_3}
    \rightsquigarrow \quad 1 + (kA' - \bigstar - \Delta_B - 1) \mod B'
\end{gather}
where $\bigstar$ is the placeholder for the term \eqref{eq::der_4}, i.e. the elapsed time since the latest network transmission (see. \figref{fig::system_extended}). As before the changes are due to the shift of B's period and the update generation delay.

Simply summing \eqref{eq::der_4} and \eqref{eq::derr_3} seems to yields an evolution for $\seq{\aa}$. However, this evolution does not account for the fact that the number $l$ of failed transmissions possibly changes over time. Hence, such a sum only yields an evolution under the assumption that for every $kA'$-th time step, the last $l$ transmissions did fail while the $(l+1)$-th transmission did succeed. We denote the sequence following such an evolution with $\seq{\aa^{[l]}}$ and point out that it might be impossible to comply with the assumption:
\begin{equation}
\begin{aligned}
    \aa^{[l]}_k &= 2 + lN' + (kA'-\Delta_N-1) \mod N'
    \\
    & 
    \begin{aligned}
        \quad +
        \Big[
        (kA'-\Delta_B-2-lN') \mod B'&
        \\
        -
        (kA'-\Delta_N-1) \mod N'&
        \Big] \mod B'
    \end{aligned}
\end{aligned}
\end{equation}
Though this looks quite more intricate than \eqref{eq::basis}, we can analyze $\seq{\aa^{[l]}}$ the same way we analyzed the basic model. Note that the values $2$ and $1$ are due to the introduced delays in update generation and transmission. It is an easy exercise to use different delay values and carry those over to any of the presented results.

\section{Results for the Basic Model}
\label{sec::extended}

Treating $\seq{\a}$ as a stochastic process, the main goal of this paper is to find usable expressions for the expectation and some upper bound on its largest values.
As it turns out, these expressions can be derived, once the following structural result on the distribution of values of $\seq{\a}$ is established: 

\begin{minipage}[t]{0.925\columnwidth}
\colorlet{shadecolor}{orange!15}\begin{shaded}
\begin{theorem}
    \label{theo::main}
    In the basic model, the period of the process $\seq{\a}$ is $aBN$. And over $aBN$ consecutive elements, the AoI sequence $\seq{\alpha}$ generates the following distribution of values:
    \begin{equation}
        \label{eq::main_theo_set}
        \ts{\alpha_k}{k = 1}{aBN}
        = \bigcup_{\substack{ i = 0\dots a-1 \\ j = 0\dots N-1 }}
        \s{c_{ij},ab,B}
    \end{equation}
    with
    \begin{multline}
        \label{eq::cij_main}
           c_{ij} = 
        iA' \mod (ab)  + {}  \\
                \left\lceil
                    \frac{iA' \mod (an)-iA' \mod (ab) + jan}{ab}
                \right\rceil
                ab 
    \end{multline}
\end{theorem}
\end{shaded}
\end{minipage}
\\

Theorem \ref{theo::main} reveals that the distribution of $\seq{\alpha}$ can be expressed as a superposition of sets of equally distanced values. This is illustrated in \figref{fig::structural_result}. The most influential quantities in this regard are the starting positions of these sets: $c_{ij}$. The ceiling and modulo operators in \eqref{eq::cij_main} lead to an intricate expression for the expected AoI which does not lend itself to practical use (derivation in the appendix):
\begin{gather}
    \label{eq::basic_exact}
    \E{\a_k} =
    \frac{B'+N'-n+ab}{2}
    - \frac{b}{N}
    \Bigg(
        \frac{a(b+1)}{2} \floor{\frac{N}{b}}
        \\ \notag
        +
        \ceil{(N \mod b) \frac{a}{b}}
        +
        \sum_{i=0}^{a-1} \sum_{j=0}^{N \mod b -1} \frac{ \left( j - \floor{\frac{ib}{a}} \right) n \mod b }{b}
    \Bigg)
\end{gather}
It can however be considerably simplified by using the central property of the ceiling operator: $x \leq \left\lceil x \right\rceil < x+1$, as the following Corollary presents:

\begin{minipage}[t]{0.925\columnwidth}
\colorlet{shadecolor}{orange!15}\begin{shaded}
\begin{corollary}
    \label{coro::main}
    In the basic model, the expected AoI is bounded by
    \begin{equation}
        \label{eq::coro1_approx}
        \E{\a_k} = \frac{B'+N' -n }{2} \pm \frac{ ab}{2}
    \end{equation}
    Using $\hE{\a_k} = \frac{B'+N' -n }{2}$ as an approximation, the maximal relative error becomes
    \begin{equation}
        \label{eq::coro1_error}
        \left\vert \frac{ \E{\a_t}-\hE{\a_t}  }{\E{\a_t}} \right\vert
        \leq \frac{ab}{(B-1)ab+n(aN-1)}
    \end{equation}
    The largest value of $\seq{\a}$ is
    \begin{equation}
        \label{eq::coro1_max}
        \max_k \, \a_k \leq B'+N'-n
    \end{equation}
\end{corollary}
\end{shaded}
\end{minipage}
\\

Proofs for both Theorem \ref{theo::main} and Corollary \ref{coro::main} are found in the appendix. Note that according to Corollary \ref{coro::main}, it is possible to increase $A$ (the part of $A'$ that is coprime to all other quantities) without changing any of the results. Furthermore, when using the center of the possible range for $\hE{\a_k}$ as an approximation for the true value $\E{\a_k}$, the relative error can become infinitely large if $B=1 = aN$. The reason lies in the fact that $\seq{\a}$ becomes a null-sequence for those parameters. Since the basic model rather serves theoretical purposes, we omit any further investigation in the quality of the approximation.

\pgfplotstableread{data_superposition_1.txt}{\superposition}
\pgfplotstableread{data_superposition_2.txt}{\superpositionn}
\begin{figure}
    \centering
    \footnotesize
    \begin{tikzpicture}
        \begin{axis}[
            width = 95mm,
            height = 55mm,
            xlabel = {AoI-Value},
            ylabel = {Occurrence},
            ]
            \addplot[ blue, mark=square* ] table [x={x1},y={y1}] {\superposition};
            \addplot[ Green, mark=* ] table [x={x2},y={y2}] {\superposition};
            \addplot[ blue, mark=square* ] table [x={x3},y={y3}] {\superposition};
            \addplot[ Green, mark=* ] table [x={x4},y={y4}] {\superposition};
            \addplot[ blue, mark=square* ] table [x={x5},y={y5}] {\superposition};
            \addplot[ red, ultra thick ] table [x={x6},y expr=\thisrow{y6}+0] {\superpositionn};
        \end{axis}
    \end{tikzpicture}
    \caption{Visualization of the structure implied by Theorem \ref{theo::main}: 5 sets of equally distanced values (connected marks), each starting with a different value. In superposition they yield the desired distribution (red). Parameters: $A=17,B=7,N=5,a=b=1,n=2$}
    \label{fig::structural_result}
\end{figure}
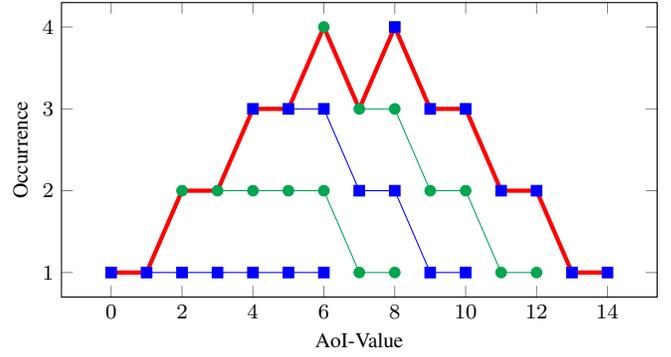

\section{Results for the Extended Model}

In analogy to the results for the basic model, we can find the same structural result for the distribution of the AoI under the extended model. Note however, that Theorem \ref{th::second} is restricted to a sequence sequence $\seq{\aa^{[l]}}$, i.e. a specific value of $l$. The parameter $l$ stands for the assumption that, no matter which processing cycle of agent A is considered, the last $l-1$ network transmissions prior to that cycle did fail. Though this might be an assumption that is impossible to fulfill, the theorem will still be crucial for following results.

\begin{minipage}[t]{0.925\columnwidth}
\colorlet{shadecolor}{orange!15}\begin{shaded}
\begin{theorem}
\label{th::second}
In the extended model, the period of the process $\seq{\aa^{[l]}}$ is $aBN$. And over $aBN$ consecutive elements, the AoI sequence $\seq{\aa^{[l]}}$ generates the following distribution of values:
\begin{equation}
    \label{eq::theorem_2_1}
    \ts{\aa^{[l]}_k}{k = 1}{aBN}
    = \bigcup_{\substack{ i = 0\dots a-1 \\ j = 0\dots N-1 }}
    \s{c_{ij}^{[l]},ab,B}
\end{equation}
with
\begin{multline}
       c_{ij}^{[l]} = 2+lN'+
    (iA'-\Delta_B-2-lN') \mod ab \ +  \\
            \left\lceil 
                \frac{\scriptstyle (iA'-\Delta_N-1) \mod an-(iA'-\Delta_B-2-lN') \mod ab + jan}{\scriptstyle ab}
            \right\rceil
            ab 
\end{multline}
\end{theorem}
\end{shaded}
\end{minipage}
\\

The similarities between Theorem \ref{theo::main} and Theorem \ref{th::second} are obvious. Notably, the change in the structure is restricted to the starting points $c_{ij}$ of the sets with equally distanced elements. As before, an exact equation for $\E{\aa_k^{[l]}}$ can be developed, but this time does not yield a much simpler structure than the one implicitly given in Theorem \ref{th::second}.
More relevant is the following corollary which, crucially, is not restricted on the impossible assumption on $\seq{\aa^{[l]}}$. The exact reason for this is provided in the proof and is based on the probabilistic occurrence of the values of $\seq{\aa^{[l]}}$ in $\seq{\aa}$.

\begin{minipage}[t]{0.925\columnwidth}
\colorlet{shadecolor}{orange!15}\begin{shaded}
\begin{corollary}
    \label{coro::second}
    Let 
    \begin{equation}
        K = 2 + (\Delta_N + 1) \cmod n
    \end{equation}
    In the extended model, the expected AoI is bounded by
    \begin{equation}
        \label{eq::coro_2_1}
        \E{\aa_t} = \frac{B'+N'-n}{2} + K
        + \frac{ \bar{p} }{p}N'
        \pm \frac{ab}{2}
    \end{equation}
    Using $\hE{\aa_t} = \frac{B'+N' -n}{2} + K + \frac{ \bar{p} }{p}N'$ as approximation, the maximal relative error becomes
    \begin{equation}
        \label{eq::coro_2_3}
        \left\vert \frac{ \E{\aa_t}-\hE{\aa_t}  }{\E{\aa_t}} \right\vert
        \leq \frac{\scriptstyle ab}{\scriptstyle (B-1)ab+n(aN-1) + 2(K + \frac{ \bar{p} }{p}N')}
    \end{equation}
    With probability $\sigma$ it holds $\forall \, k$ that
    \begin{equation}
        \label{eq::coro_2_2}
        \aa_k
        \leq
        B'+N'-n + K + N'\left\lceil \frac{\ln(1-\sigma)}{\ln(1-p)} - 1 \right\rceil
    \end{equation}
\end{corollary}
\end{shaded}
\end{minipage}
\\

\section{Approximation Error}

Because the exact equation for $\E{\aa_k}$ is too intricate to analyze analytically, Figs. \ref{fig::error_ABN}, \ref{fig::error_abn} and \ref{fig::error_p} illustrate the distribution of the relative approximation error (with sign) $\frac{ \E{\aa_t}-\hE{\aa_t} }{\E{\aa_t}}$ for the extended model. For this investigation, we calculated both the exact and approximated value of $\E{\aa_k}$ for roughly $250 000$ different parameter configurations, where the following parameters are changed: $A',B',N',a,b,n,\Delta_B,\Delta_N,p$.
Note though, that parameter configurations who fulfill $A'=B'=N'$ are excluded; their approximation error is very high but can also be easily avoided since it is trivial to derive the AoI sequence in those cases. Technically, the illustrated distributions are histograms with finite bin sizes. It is however more comprehensive to illustrate the distributions with a closed line. Since the area under any probability distribution must be 1, values on the y-axis are omitted.

It is evident, that the average approximation error is not zero but rather about $-4\%$. I.e. one can expect our approximation to yield slightly higher average AoI values than will occur in reality. Furthermore, almost all relative approximation errors are contained within the bound $-4 \% \pm 10 \%$.

\pgfplotstableread{data/data_Ad.txt}{\dataAd}
\pgfplotstableread{data/data_Bd.txt}{\dataBd}
\pgfplotstableread{data/data_Nd.txt}{\dataNd}
\pgfplotstableread{data/data_a.txt}{\dataa}
\pgfplotstableread{data/data_b.txt}{\datab}
\pgfplotstableread{data/data_n.txt}{\datan}
\pgfplotstableread{data/data_dB.txt}{\datadB}
\pgfplotstableread{data/data_dN.txt}{\datadN}
\pgfplotstableread{data/data_p.txt}{\datap}

\pgfplotsset{style a/.style={
    scale only axis,
    width = 80mm,
    height = 28mm,
    xlabel = {Relative error made by approximation $\hat{\mathbb{E}}[\aa_t]$},
    ylabel = {Occurence},
    ytick style = {draw=none},
    xtick style = {draw=none},
    ytick={0},
    yticklabels={},
    ymajorgrids={true},
    xmajorgrids={true},
    xtick={-1,-0.5,-0.15,0,0.15,0.5,1},
    cycle list={ {blue}, {red}, {Green} }
}}

\pgfplotsset{style b/.style={
    scale only axis,
    width = 80mm,
    height = 28mm,
    ylabel = {Occurence},
    ytick={0},
    yticklabels={},
    ymajorgrids={true},
    ytick style = {draw=none},
    xmajorgrids={true},
    xtick={-1,-0.5,-0.15,0,0.15,0.5,1},
    xticklabels={},
    xtick style = {draw=none},
    cycle list={ {blue}, {red}, {Green} }
}}

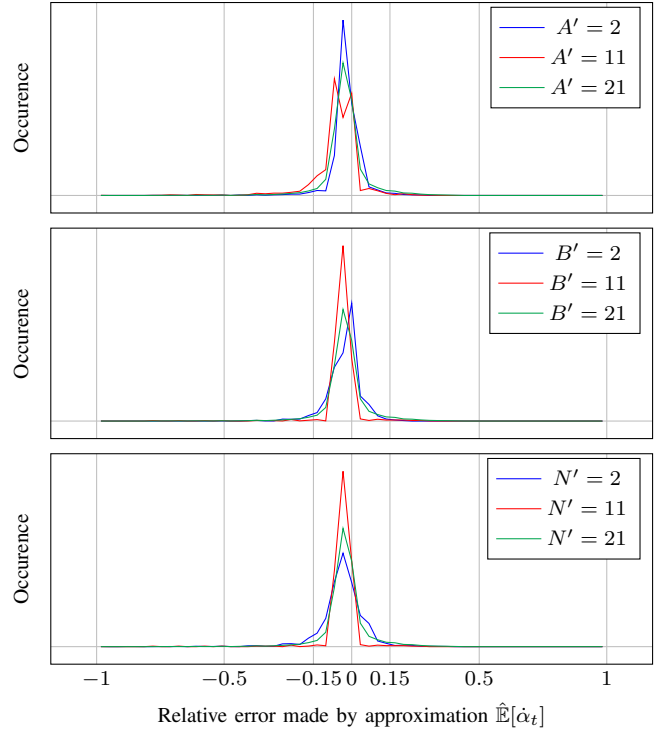
\begin{figure}[H]
    \centering
    \footnotesize
    \begin{tikzpicture}
        \begin{axis}[style b]
            \addplot table [x={x},y={ymin}] {\dataAd};
            \addplot table [x={x},y={ymid}] {\dataAd};
            \addplot table [x={x},y={ymax}] {\dataAd};
            \legend{$A' = 2$,$A'=11$,$A'=21$};
        \end{axis};
    \end{tikzpicture}
    \begin{tikzpicture}
        \begin{axis}[style b]
            \addplot table [x={x},y={ymin}] {\dataBd};
            \addplot table [x={x},y={ymid}] {\dataBd};
            \addplot table [x={x},y={ymax}] {\dataBd};
            \legend{$B' = 2$,$B'=11$,$B'=21$};
        \end{axis};
    \end{tikzpicture}
    \begin{tikzpicture}
        \begin{axis}[style a]
            \addplot table [x={x},y={ymin}] {\dataNd};
            \addplot table [x={x},y={ymid}] {\dataNd};
            \addplot table [x={x},y={ymax}] {\dataNd};
            \legend{$N' = 2$,$N'=11$,$N'=21$};
        \end{axis};
    \end{tikzpicture}
    \caption{Distribution of the relative approximation error $\frac{ \E{\aa_t}-\hE{\aa_t} }{\E{\aa_t}}$ for fixed values of $A'$,$B'$ and $N'$. Values $2$ and $21$ are smallest and largest value assigned for $A'$,$B'$ and $N'$ in all parameter configurations.}
    \label{fig::error_ABN}
\end{figure}    
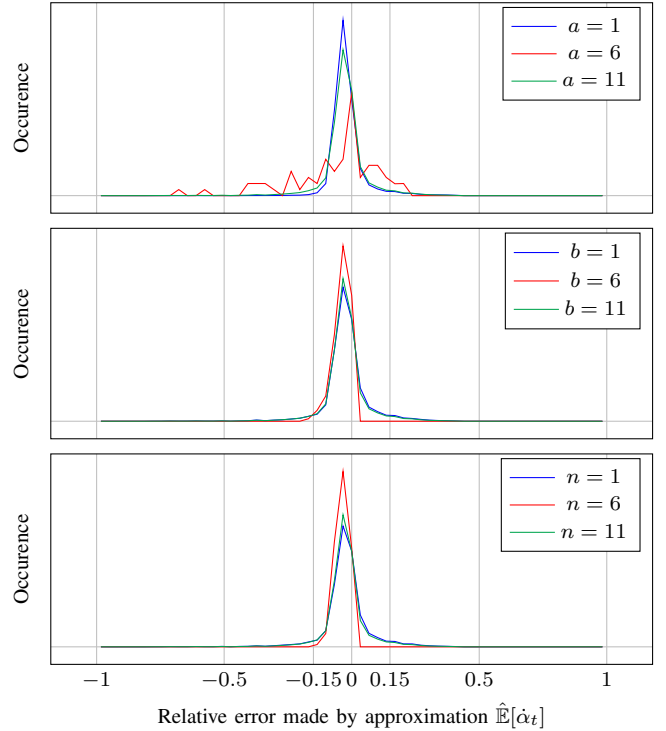
\begin{figure}[H]
    \centering
    \footnotesize    
    \begin{tikzpicture}
        \begin{axis}[style b]
            \addplot table [x={x},y={ymin}] {\dataa};
            \addplot table [x={x},y={ymid}] {\dataa};
            \addplot table [x={x},y={ymax}] {\dataa};
            \legend{$a = 1$,$a=6$,$a=11$};
        \end{axis};
    \end{tikzpicture}
    \begin{tikzpicture}
        \begin{axis}[style b]
            \addplot table [x={x},y={ymin}] {\datab};
            \addplot table [x={x},y={ymid}] {\datab};
            \addplot table [x={x},y={ymax}] {\datab};
            \legend{$b = 1$,$b=6$,$b=11$};
        \end{axis};
    \end{tikzpicture}
    \begin{tikzpicture}
        \begin{axis}[style a]
            \addplot table [x={x},y={ymin}] {\datan};
            \addplot table [x={x},y={ymid}] {\datan};
            \addplot table [x={x},y={ymax}] {\datan};
            \legend{$n = 1$,$n=6$,$n=11$};
        \end{axis};
    \end{tikzpicture}
    \caption{Distribution of the relative approximation error $\frac{ \E{\aa_t}-\hE{\aa_t} }{\E{\aa_t}}$ for fixed values of $a$,$b$ and $n$. Values $1$ and $11$ are smallest and largest value assigned for $a$,$b$ and $n$ in all parameter configurations.}
    \label{fig::error_abn}
\end{figure}    
\begin{figure}[H]
    \centering
    \footnotesize
    \begin{tikzpicture}
        \begin{axis}[style b]
            \addplot table [x={x},y={ymin}] {\datadB};
            \addplot table [x={x},y={ymid}] {\datadB};
            \addplot table [x={x},y={ymax}] {\datadB};
            \legend{$\Delta_B = 0$,$\Delta_B=10$,$\Delta_B=20$};
        \end{axis};
    \end{tikzpicture}
    \begin{tikzpicture}
        \begin{axis}[style b]
            \addplot table [x={x},y={ymin}] {\datadN};
            \addplot table [x={x},y={ymid}] {\datadN};
            \addplot table [x={x},y={ymax}] {\datadN};
            \legend{$\Delta_N = 0$,$\Delta_N=10$,$\Delta_N=20$};
        \end{axis};
    \end{tikzpicture}
    \begin{tikzpicture}
        \begin{axis}[style a]
            \addplot table [x={x},y={ymin}] {\datap};
            \addplot table [x={x},y={ymid}] {\datap};
            \addplot table [x={x},y={ymax}] {\datap};
            \legend{$p = 0.1$,$p=0.5$,$p=1$};
        \end{axis};
    \end{tikzpicture}
    \caption{Distribution of the relative approximation error $\frac{ \E{\aa_t}-\hE{\aa_t} }{\E{\aa_t}}$ for fixed values of $\Delta_B$,$\Delta_N$ and $p$. Values $0$ and $20$ are smallest and largest value assigned for $\Delta_B$ and $\Delta_N$ in all parameter configurations.}
    \label{fig::error_p}
\end{figure}
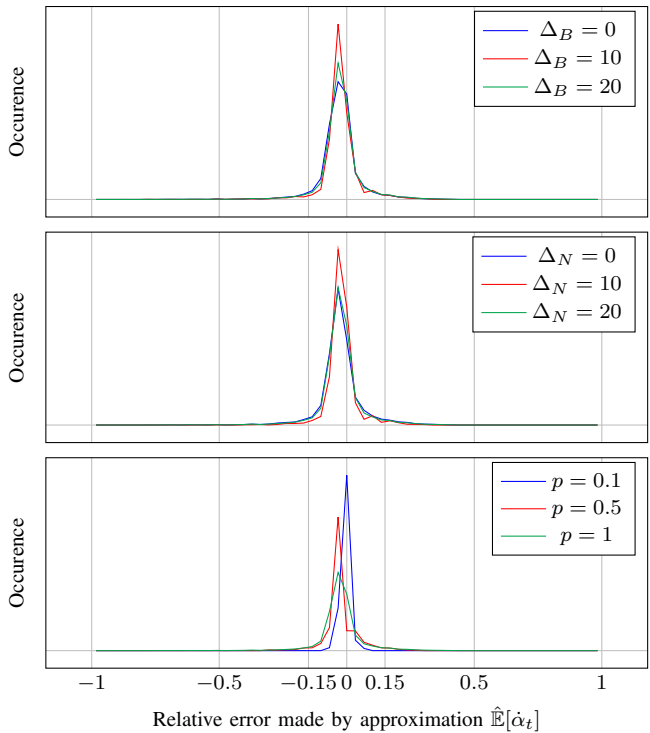

\section{Conclusion}

For the described system (see \figref{fig::general_system}),
we derived a structural result on the distribution of the values of the AoI process. Based on that result we developed both an exact and an approximating expression for the mean of the process together with a probabilistic bound on the maximal occurring AoI. The simulations show that the approximation is good enough to enable efficient calculation of the mean AoI with limited error.

\bibliographystyle{IEEEtran}
\bibliography{library}

\section{Appendix}

To present the proofs as succinct as possible, we will often use the notation $\mathcal{Y} = \mathcal{X} + x$ to express that the set $\mathcal{Y}$ is generated by adding the value $x$ to every element of the set $\mathcal{X}$.

\begin{lemma}
    \label{lemma::seq_generation}
    Let $X,Y \in \mathbb{N}$ and $\gcd{(X,Y)} = 1$ (i.e. $X$ and $Y$ are coprime). Then 
    \begin{equation}
        \label{eq::gen_set_1}
        \ts{kX \mod Y}{k=1}{Y}
        = \s{0,1,Y} 
    \end{equation}
\end{lemma}
\begin{proof}
    We proof via contradiction. Denote the following residues, starting with arbitrary $k \in \mathbb{N}$ as
    \begin{equation}
        \begin{aligned}
            kX \mod Y &= z_0
            \\
            (k+1)X \mod Y &= z_1
            \\
            &\vdots
            \\
            (k+Y-1)X \mod Y &= z_{Y-1}
            \\
            (k+Y)X \mod Y &= kX \mod Y = z_0
        \end{aligned}
    \end{equation}
    with $z_i \in \{0,1, \dots Y-1\}$.
    Suppose that there are at least two numbers from $\{z_0,z_1,\dots z_{Y-1}\}$ which are equal, say $z_{k'}$ and $z_{k''}$, then
    \begin{equation}
    \begin{gathered}
        (k+k')X \mod Y = (k+k'')X \mod Y
        \\
        \Rightarrow (k+k')X - (k+k'')X = (k'-k'')X = wY
    \end{gathered}
    \end{equation}
    for some $w \in \mathbb{N}$. Per construction we have $|k'-k''| < Y$ which implies that $\lcm(X,Y) = |k'-k''|X < YX$. This must be wrong since $X$ and $Y$ are coprime, proving \eqref{eq::gen_set_1}.
\end{proof}

\begin{lemma}
    \label{lemma::meetings}
    Let $\seq{x}$ be a sequence with period $X$ and denote its values according to $(x_{1+lX},\dots x_{X+lX}) = (\r{x}_1,\dots \r{x}_X)$ for all $l\in\mathbb{N}$. Likewise let $\seq{y}$ be a sequence with period $Y$ and $(y_{1+lY},\dots y_{Y+lY}) = (\r{y}_1,\dots \r{y}_Y)$. Further let $z = \gcd(X, Y)$. Then over any $\frac{XY}{z}$ consecutive steps $k$, the set of generated pairs $(x_k,y_k)$ contains exactly each pairing $(\r{x}_i,\r{y}_j)$ once, whose indices are part of the same residual class with regard to $z$:
    \begin{equation}
        \label{eq::lemma_sequence_pairs}
        \ts{(x_k,y_k) }{ k = 1}{\frac{XY}{z} }
        =
        \bigcup_{r = 1}^{z}
        \bigcup_{\substack{ i \in \s{r,z,\frac{X}{z}} 
            \\ j \in \s{r,z,\frac{Y}{z} } } }
        \left\{
            \vphantom{M^{M^M}_M} (\r{x}_i,\r{y}_j)
        \right\}
    \end{equation}
\end{lemma}
\begin{proof}
    For $z=1$, \eqref{eq::lemma_sequence_pairs} means that every possible pairing occurs exactly once.
    Conversely, assume that a specific pairing $(\r{x}_i,\r{y}_j)$ is generated twice during the $XY$ steps and denote these steps with $k'$ and $k''$. Then $|k'' - k'| = \Delta < XY$ and, due to the periodicity, both $X$ and $Y$ must be divisors of $\Delta$, making $\Delta$ a common multiple smaller $XY$. Since $X$ and $Y$ are coprime this is a contradiction.
    
    For $z \in \mathbb{N}$, construct two new sequences $(\t{x}_\kappa)$ and $(\t{y}_\kappa)$ according to
    \begin{equation}
        \label{eq::regourping}
    \begin{aligned}
        \t{x}_\kappa &= ( {x}_{\kappa z+1},\dots {x}_{\kappa z+z} )
        \\
        \t{y}_\kappa &= ( {y}_{\kappa z+1},\dots {y}_{\kappa z+z} )
    \end{aligned}
    \end{equation}
    I.e. one element of $(\t{x}_\kappa)$ consists of $z$ consecutive elements from $(x_k)$.
    Per construction $(\t{x}_\kappa)$ has periodicity $\frac{X}{z}$ and $(\t{y}_\kappa)$ has periodicity $\frac{Y}{z}$. Also, $\frac{X}{z}$ and $\frac{Y}{z}$ are coprime.
    Hence, according to the first part of the proof, over any $\frac{X}{z}\cdot \frac{Y}{z}$ consecutive steps, each of the $\frac{X}{z}$ different elements of $(\t{x}_\kappa)$ will eventually pair with each of the $\frac{Y}{z}$ different elements of $(\t{y}_\kappa)$.
    Using \eqref{eq::regourping} this implies that over any $\frac{XY}{z^2}\cdot z$ consecutive steps of the original sequence, each $x_{i z + \kappa}$ meets every $y_{j z + \kappa}$ exactly once, with $i = 1,\dots \frac{X}{z} ; \ j = 1,\dots \frac{Y}{z}; \ \kappa = 1\dots z$. With slight redefinition of $i$ and $j$ (such that they include $z$) this means that for given $\kappa$, $\r{x}_{i+\kappa}$ meats $\r{y}_{j+\kappa}$ iff $i \in \{0,z,2z,\dots \left(\frac{X}{z}-1\right) z \}$ and $j \in \{0,z,2z,\dots \left(\frac{Y}{z}-1\right)z\}$. The proposition follows.
\end{proof}

\begin{lemma}
    \label{lemma::set_equation}
    Let $X,x,y \in \mathbb{N}$, then
    \begin{align}
        \label{eq::lemma_3_1}
        \s{y,x,X} \mod Xx
        &=
        \s{ y \mod x ,x,X}
        \\
        \label{eq::lemma_3_2}
        \s{-y,x,X} \mod Xx
        &=
        \s{ y \cmod x ,x,X}
    \end{align}
\end{lemma}
\begin{proof}
    We proof \eqref{eq::lemma_3_2} and initially assume that $y < Xx$.
    Then $\s{-y,x,X}$ will consist of positive and possibly negative values. Continuously adding $x$ to $-y$ eventually starts producing non-negative results, beginning with the  $\left\lceil\frac{y}{x} \right\rceil$-th summation step. This motivates the separation
    \begin{equation}
    \begin{gathered}
        \s{-y,x,X}
        \\
            =
            \underbrace{
            \s{-y,x,\left\lfloor \frac{y}{x} \right\rfloor}
            }_{\displaystyle \text{negative values}}
            \cup
            \underbrace{
            \s{-y+\left\lceil \frac{y}{x} \right\rceil x ,x,X-\left\lfloor \frac{y}{x} \right\rfloor}
            }_{\displaystyle \text{non-negative values}}
    \end{gathered}
    \end{equation}
    The largest term of the "non-negative" set is $-y+(X-1)x$. Therefore, taking all values modulo $(Xx)$ only effects the "negative" set by adding $Xx$. Doing so, the smallest term in the "negative" set becomes $-y + Xx$ (which is exactly one $x$ greater than the largest term of the "non-negative" set). Hence, after the ${} \mod (Xx)$ operation, the elements of both sets can be rearranged into a new set according to
    \begin{equation}
    \begin{gathered}
        \left( \s{0,x,X} - y \right) \mod (Xx)
        \\
            =
            \s{-y+Xx,x,\left\lfloor \frac{y}{x} \right\rfloor}
            \cup
            \s{-y+\left\lceil \frac{y}{x} \right\rceil x ,x,X-\left\lfloor \frac{y}{x} \right\rfloor}
        \\
            =
            \s{-y+\left\lceil \frac{y}{x} \right\rceil x ,x,X}
    \end{gathered}
    \end{equation}
    
    Suppose now that $y \geq Xx$ and in particular $y \mod Xx = z$. Then we get
    \begin{equation}
    \begin{aligned}
        \s{-y,x,X} \mod Xx
        &= \s{-y \mod Xx,x,X} \mod Xx
        \\
        &= \s{-z,x,X} \mod Xx
        \\
        &\overset{*}{=} \s{-z+\left\lceil \frac{z}{x} \right\rceil x ,x,X}
        \\
        &\overset{**}{=} \s{-y+\left\lceil \frac{y}{x} \right\rceil x ,x,X}
    \end{aligned}
    \end{equation}
    where (*) is true because $z < Xx$ and (**) is readily verified by substitution ($y = wXx+z$ for some $w \in \mathbb{N}$). To proof \eqref{eq::lemma_3_1}, see that
    \begin{equation}
        (-y) \mod x = (-y) - \left\lfloor \frac{(-y)}{x} \right\rfloor x = -y +\left\lceil \frac{y}{x} \right\rceil x = y \cmod x
    \end{equation}
    and hence substitution of $-y$ for $y$ in \eqref{eq::lemma_3_2} yields \eqref{eq::lemma_3_1}.
\end{proof}


\noindent
\textbf{Proof to Theorem 1}
\\ \noindent
The evolution of the AoI sequence $\seq{\alpha}$ according to \eqref{eq::basis} can be interpreted as a function of the sequences $\seq{\eta}$ and $\seq{\beta}$, identified by:
\begin{multline}
    \label{eq::evo_2}
    \a_k = \a_k(\beta_k, \eta_k) = 
    \n_k + (\b_k - \n_k) \mod B' =
    \\
                        \underbrace{ \vphantom{\frac{1}{2}}
    kA' \mod N'         
                        }_{\displaystyle =: \eta_k}
    +  \big( 
                        \underbrace{ \vphantom{\frac{1}{2}}
    kA' \mod B'         
                        }_{\displaystyle =: \beta_k}
    -                   \underbrace{ \vphantom{\frac{1}{2}}
    kA' \mod N'      
                        }_{\displaystyle =: \eta_k}
    \big) \mod B'
\end{multline}

This means that
\begin{equation}
\begin{aligned}
    \b_k &= b \big[ k(nA) \mod (aB)  \big]
    \\
    \n_k &= n \big[ k(bA) \mod (aN)  \big]
\end{aligned}
\end{equation}
and we know from Lemma \ref{lemma::seq_generation} that $\seq{\beta}$ and $\seq{\eta}$ do generate the sets $\s{0,b,aB}$ and $\s{0,n,aN}$ (over their respective periods). With one element from each set, we can readily describe the resulting value for $\a_k$. However, not all possible element pairings are realized at the same step $k$. Our first priority is therefore to determine which elements (and values) are realized simultaneously.
To that end it is prudent to first denote the elements in appearing order as
\begin{equation}
    \begin{aligned}
        \b_{1+l(aB)} &=: \r{b}_1 \\
        \b_{2+l(aB)} &=: \r{b}_2 \\
        & \vdots \\
        \b_{aB +l(aB)} &=: \r{b}_{aB}
    \end{aligned}
    \qquad
    \begin{aligned}
        \n_{1+l(aN)} &=: \r{n}_1 \\
        \n_{2+l(aN)} &=: \r{n}_2 \\
        & \vdots \\
        \n_{aN +l(aN)} &=: \r{n}_{aN}
    \end{aligned}
\end{equation}
with $l\in\mathbb{N}$.

Since the periods of $\seq{\b}$ and $\seq{\n}$ are $aB$ and $aN$, the period of $\seq{\a}$ must be $aBN$.
From Lemma \ref{lemma::meetings}, we know that during those $aBN$ steps only those $\r{b}_i$ and $\r{n}_j$ will occur at the same step, whose indices belong to the same residue class $r\in \mathbb{N}$ with respect to the greatest common divisor $a$:
\begin{equation}
\begin{gathered}
    \ts{\a_k(\b_k,\n_k)}{k=1}{aBN}
    =
    \bigcup_{r=0}^{a-1} \bigcup_{i=1}^{B} \bigcup_{j=1}^N \big\{ \a_k(\r{b}_{ia+r}, \r{n}_{ja+r}) \big\}
\end{gathered}
\end{equation}
Looking at the value behind $\r{b}_{ia+r}$ we have
\begin{equation}
\begin{aligned}
    \r{b}_{ia+r}
    &= (ia+r)A' \mod B'
    \\&
    = (iaA' \mod B' + rA') \mod B'
    \\&
    = (iabnA \mod abB + rA') \mod B'
    \\&
    = \left( ab (inA \mod B) + rA' \right) \mod B'
\end{aligned}
\end{equation}
Over all possible $i$, i.e. for $i=1,\dots B$, the term $ab (inA \mod B)$ generates the set $\s{0,ab,B}$. And therefore, using lemma \ref{lemma::set_equation}
\begin{equation}
    \label{eq::values_for_b}
    \begin{aligned}
    \ts{\r{b}_{ia+r}}{i=1}{B}
    &=
    \left( \s{0,ab,B} + rA' \right) \mod abB 
    \\
    &=
    \s{rA' \mod ab, ab, B}
    \end{aligned}
\end{equation}
Of course, the same holds true for $\r{n}_{ja+r}$:
\begin{equation}
    \label{eq::values_for_n}
    \ts{\r{n}_{ja+r}}{j=0}{N}
    =
    \s{rA' \mod an, an, N}
\end{equation}

Knowing that we can look back at our starting equation:
\begin{equation}
    \a_k = \n_k + (\b_k - \n_k) \mod B'
\end{equation}
Over $\seq{\a}$'s period $aBN$, every value from \eqref{eq::values_for_b} will pair up with every value from \eqref{eq::values_for_n} exactly once. I.e. for the value $\r{n}_{ja+r}$ from \eqref{eq::values_for_n} the sequence $\seq{\a}$ generates the values
\begin{equation}
    \r{n}_{ja+r} + 
    \left( 
    \s{rA' \mod ab, ab, B}
    -
    \r{n}_{ja+r}
    \right) \mod B'
\end{equation}
which according to lemma \ref{lemma::set_equation} is equal to
\begin{equation}
    \label{eq::partial}
    \s{
    rA' \mod ab +
    \left\lceil
    \frac{\r{n}_{ja+r} - rA' \mod ab}{ab}
    \right\rceil ab
    ,ab,B
    }
\end{equation}
To yield an expression for all generated values, we have to take the union of \eqref{eq::partial} over all possible $\r{n}_{ja+r}$. I.e. the union over all possible residue classes $r = 0,\dots a-1$, and per residue class over all $j=0,\dots N-1$. Since we do not care for the order of the generated values, we can invoke \eqref{eq::values_for_n} for the union over $j$ and obtain
\begin{equation}
    \label{eq::partial2}
    \begin{gathered}
    \bigcup_{j=0}^{N-1}
    \s{ c_{rj} ,ab,B}
    \\
    c_{rj} = 
    rA' \mod ab +
    \left\lceil
    \frac{rA' \mod an + jan - rA' \mod ab}{ab}
    \right\rceil ab
    \end{gathered}
\end{equation}
Replacing $r$ with $i$ and taking the union over $i = 0,\dots a-1$ yields the proposition. \qed
\\

\noindent
\textbf{Derivation of Equation \eqref{eq::basic_exact}}
\\
To obtain the expected AoI, we need an expression for the expected value of $c_{ij}$, because
\begin{equation}
    \label{eq::avg_aoi}
\begin{aligned}
    \E{\a_k}
    &=
    \frac{1}{aBN} \sum_{k=1}^{aBN} \a_k
    =
    \frac{1}{aBN}
    \sum_{\substack{ i = 0\dots a-1 \\ j = 0\dots N-1 }}
    \sum_{x \in \s{c_{ij},ab,B}} x
    \\
    &=
    \frac{1}{aN}
        \sum_{\substack{ i = 0\dots a-1 \\ j = 0\dots N-1 }}
        c_{ij}
     + \frac{(B-1)ab}{2}
\end{aligned}
\end{equation}
Looking at $c_{ij}$ as defined in \eqref{eq::cij_main}, the term containing the ceiling operator can be developed in the following way if we consider its sum over $i = 0,\dots a-1$ \textit{and} $j=0,\dots N-1$:
\begin{align}
        & \mathrel{\phantom{=}} \notag
        \sum \left\lceil
        \frac{iA' \mod an + jan - iA' \mod ab}{ab}
        \right\rceil
    \\
        & \overset{(1)}{=} \notag
        \sum \ceil{
        \frac{n \left( ib - \floor{\frac{ib}{a}}a \right) + jan - b \left( in - \floor{\frac{in}{a}}a \right)}{ab}
        }
    \\
        &= \notag
        \sum \floor{\frac{in}{a} } +
        \sum \ceil{ \frac{
        \left( j - \floor{\frac{ib}{a}} \right) n
        }{b} 
        }
\end{align}
Equality $(1)$ holds because of \eqref{eq::compositions} and \eqref{eq::mod_to_ceil} and the fact that the factor $A$ merely reorders the occurrence of the rest classes and can therefore be omitted due to the sum (see Lemma \ref{lemma::seq_generation}).

Applying the substitution
\begin{equation}
    \frac{x}{y} = \underbrace{ \frac{x - x \mod y}{y} }_{\displaystyle \in \mathbb{N}} + \frac{x \mod y}{y}
\end{equation}
to the remaining ceiling and floor operators makes it possible to retract the integer part from the operators.
The resulting sums are readily evaluated except the following two:
\begin{multline}
    \label{eq::first_hard}
    \sum
    \ceil{
    \frac{
    \left( j - \floor{\frac{ib}{a}} \right) n \mod b
    }{b}
    }
    \\
    =
    aN - a \floor{ \frac{N}{b} } - \ceil{ (N \mod b) \frac{a}{b} } 
\end{multline}
and
\begin{multline}
    \label{eq::second_hard}
    \sum \frac{\left( j - \floor{\frac{ib}{a}} \right) n \mod b}{b}
    =
    \frac{a(b-1)}{2} \floor{\frac{N}{b}}
    \\
    +
    \sum_{i=0}^{a-1} \sum_{j=0}^{N\mod b -1} \frac{ \left( j - \floor{\frac{ib}{a}} \right) n \mod b }{b}
\end{multline}
Equation \eqref{eq::first_hard} follows since the ceiling operator always evaluates to $1$ except when $j-\floor{\frac{ib}{a}}$ is a multiple of $b$. In those cases it evaluates to $0$. For $i=0$ this happens $\floor{\frac{N}{b}}+1$ times. This does not change as long as $\floor{\frac{ib}{a}} < (N-1) \mod b$, i.e. $\ceil{ \left( N \mod b \right) \frac{a}{b} }$ times. For all other cases a multiple of $b$ is only realized $\floor{\frac{N}{b}}$ times.

Equation \eqref{eq::second_hard} follows because any consecutive $b$ summands over $j$ evaluate to $\frac{b-1}{2}$. And there are only $\floor{\frac{N}{b}}$ such sequences of summands over $j$.

Rejoining all of these results, one can yields the following expression for the sum over all $c_{ij}$:
\begin{equation}
    \label{eq::some}
    \begin{gathered}
    \frac{1}{aN} \sum_{\substack{i=0,\dots a-1 \\ j = 0,\dots N-1}} c_{ij}
    =
    \frac{N'-n+2ab}{2}
    - \frac{b}{N}
    \Bigg(
        \frac{a(b+1)}{2} \floor{\frac{N}{b}}
        \\
        +
        \ceil{(N \mod b) \frac{a}{b}}
        +
        \frac{1}{b} \sum_{i=0}^{a-1} \sum_{j=0}^{N \mod b -1} \left( j - \floor{\frac{ib}{a}} \right) n \mod b
    \Bigg)
    \end{gathered}
\end{equation}
Finally \eqref{eq::basic_exact} results by substituting \eqref{eq::some} into \eqref{eq::avg_aoi}. \qed
\\

\noindent
\textbf{Proof to Corollary \ref{coro::main}}
\\
To obtain the approximation \eqref{eq::coro1_approx}, we take \eqref{eq::cij_main} at face value and use the bounds of the ceiling operator. This gives
\begin{equation}
    \label{eq::basis_bound}
    c_{ij}
    =
    iA' \mod (an) + jan + \frac{ab \pm ab}{2}
\end{equation}
Applying Lemma \ref{lemma::seq_generation} yields
\begin{equation}
    \frac{1}{aN}\sum_{ij} c_{ij}
    =
    \frac{n(aN-1)+ab \pm ab}{2}
\end{equation}
Substituting this into \eqref{eq::avg_aoi} finalizes the derivation.

For \eqref{eq::coro1_error}, note that once $\hE{\a_k} = \frac{B'+N' -n}{2}$ is set, the largest relative error occurs if $\E{\a_k}$ is on the lower boundary of the possible set, i.e. if $\E{\a_k} = \frac{B'+N'-n-ab}{2}$. The distance between $\hE{\a_k}$ and $\E{\a_k}$ is then $\frac{ab}{2}$. Division yields the proposition.

The largest value of $\seq{\a}$ depends on the largest value of $c_{ij}$ which again can be bounded via the bounds of the ceiling operator:
\begin{equation*}
    \begin{aligned}
        \max_k \ \a_k 
            & = \max_{{ij}} \ c_{ij} + (B-1)ab
            \\
            & = \max_i \ \{ iA'\mod an \} +  \max_j \ \{ jan \} + B' 
            \\
            & = n(a-1) + (N-1)an + B' = B' + N' - n
    \end{aligned}
\end{equation*}
which proves corollary. \qed
\\

\noindent
\textbf{Proof to Corollary \ref{coro::second}}
\\
    We separate the values generated by $\seq{\aa}$ in $aBN$-long groups of elements. Then the expectation can be expressed via
    \begin{equation}
    \label{eq::expectation}
    \begin{aligned}
        \E{\aa_k}
        &= \lim_{K \to \infty} \frac{1}{K} \sum_{k=0}^K \aa_k
        \\
        &= \lim_{K \to \infty} \frac{1}{K} \sum_{k'=0}^K \frac{1}{aBN} \sum_{k=0}^{aBN-1} \aa_{k'aBN+k}
        \\
        &= \frac{1}{aBN} \sum_{k=0}^{aBN-1} \lim_{K \to \infty} \frac{1}{K} \sum_{k'=0}^K  \aa_{k'aBN+k}
        \\
        &= \frac{1}{aBN} \sum_{k=0}^{aBN-1} p \aa_k^{[0]} + p\bar{p} \aa_k^{[1]} + p\bar{p}^2 \aa_k^{[2]} + \dots
    \end{aligned}
    \end{equation}
    Notice that this way, we circumvented the "impossibility" of the assumption for $\seq{\aa^{[l]}}$,
    since with probability of $p\bar{p}^l$ it is indeed the case that, relative to time step $t=kA'$, the last $l$ transmissions did fail and the $(l+1)$-th did succeed, for each of the $aBN$ elements. 
    According to Theorem \ref{th::second} and Lemma \ref{lemma::set_equation} we get
    \begin{gather}
    \begin{aligned}
        \frac{1}{aN} \sum_{ \substack{i=0,\dots a-1 \\ j = 0,\dots N-1 }}        
        c_{ij}^{[l]}
        =
        \frac{1}{aN} \sum_{ \substack{i=0,\dots a-1 \\ j = 0,\dots N-1 }}     
        2 + lN  + jan + \frac{ab \pm ab}{2} &&
        \\
        \ \ + (iA'-\Delta_N-1) \mod an &&
    \end{aligned} \notag
        \\
        =
        2 + lN + \frac{(N'-n)}{2}  + \frac{ab \pm ab}{2} + (\Delta_N + 1) \cmod n
    \end{gather}
    because $(iA'-\Delta_N-1) \mod an = (iA' \mod an -\Delta_N-1) \mod an$ which, over $i=0,\dots a-1$, generates the set $(\s{0,n,a} - -\Delta_N-1) \mod an$.
    Hence, with $K=2+ (\Delta_N + 1) \cmod n$,
    \begin{equation}
    \begin{aligned}
        \frac{1}{aBN}
        \sum_{k=0}^{aBN} \aa_k^{[l]}
        &=
        \frac{1}{aBN}
        \sum_{ \substack{i=0,\dots a-1 \\ j = 0,\dots N-1 }}        
        \s{c_{ij}^{[l]},ab,B}
        \\
        &= ab\frac{(B-1)}{2} + \frac{1}{aN} \sum_{ \substack{i=0,\dots a-1 \\ j = 0,\dots N-1 }} c_{ij}^{[l]}
        \\
        &= \frac{B'+N'-n}{2} + K 
        + lN  \pm \frac{ab}{2}
    \end{aligned}
    \end{equation}
    Substituting this back into \eqref{eq::expectation} and using the well-known results on polylogarithms yields
    \begin{equation}
    \begin{aligned}
        \E{\aa_k} &= \frac{B'+N'-n}{2} + K \pm \frac{ab}{2}
        + p \lim_{L \to \infty} \sum_{l=0}^L \bar{p}^l lN
        \\&=
        \frac{B'+N'-n}{2} + K \pm \frac{ab}{2}
        + \frac{\bar{p}}{p}N
    \end{aligned}
    \end{equation}
    This proves proposition \eqref{eq::coro_2_1}.
    
    As for proposition \eqref{eq::coro_2_2}, notice that (polylogarithm)
    \begin{equation}
        p\left( 1 + \bar{p} + \bar{p}^2 + \dots + \bar{p}^{\lambda -1} \right)
        =
        1- \bar{p}^\lambda
    \end{equation}
    I.e. considering only the first $\lambda$ summands in \eqref{eq::expectation} means considering only $1 - \bar{p}^\lambda$ percent of all possibly occurring values. The largest value among these belongs to the $\bar{p}^{\lambda-1}$ part and is easily evaluated to be
    \begin{equation}
        \max_{ij} c_{ij}^{[\lambda-1]} + (B-1) ab
        \leq
        B'+N' - n + K + (\lambda -1)N
    \end{equation}
    Setting $\sigma=1 - \bar{p}^{\lambda}$ yields proposition \eqref{eq::coro_2_2}.

    Finally, \eqref{eq::coro_2_3} follows for the same reasons as the corresponding proposition from Corollary \ref{coro::main}. \qed

\end{document}